\documentclass[a4paper, twoside]{amsart}
\usepackage{mathrsfs, tabularx}
\usepackage {amsthm, epsfig, amsmath, amssymb, amscd, float, latexsym}

\theoremstyle{plain}
\newtheorem{theorem}{Theorem}[section]
\newtheorem{proposition}[theorem]{Proposition}

\theoremstyle{definition}
\newtheorem{definition}[theorem]{Definition}

\newtheorem{example}[theorem]{Example}

\newcommand {\Set}[1] {\mathbb{#1}}
\newcommand{\setR}[0]{\Set{R}}

\newcommand{\setC}[0]{\Set{C}}

\newcommand{\cF}[0]{{\mathscr{F}}}

\newcommand{\slaz}[0]{{\setminus\{0\}}}

\newcommand{\cG}[0]{{\mathscr{G}}}

\newcommand {\proofBox}[0]{\hfill $\Box$ }

\newcommand {\proofread}[1]{ }

\newcommand{\pd}[2]{\frac{\partial #1}{\partial #2}}

\newcounter{saveenum}

\title{Non-dissipative electromagnetic medium with a double light cone}
\author[Dahl]{Matias F. Dahl}
\address{
Matias Dahl,
Aalto University,
Mathematics,
P.O. Box 11100,
FI-00076 Aalto,
Finland
}

\urladdr{
http://www.math.tkk.fi/\textasciitilde{}fdahl/}

\date{\today}

\begin{document}

\begin{abstract}
  We study Maxwell's equations on a 4-manifold where the
  electromagnetic medium is modelled by an antisymmetric $2\choose
  2$-tensor with real coefficients.  In this setting the \emph{Fresnel
    surface} is a fourth order polynomial surface in each cotangent
  space that acts as a generalisation of the light cone determined by
  a Lorentz metric; the Fresnel surface parameterises electromagnetic
  wave-speeds as a function of direction.  The contribution of this
  paper is a pointwise description of all electromagnetic
  medium tensors that satisfy the following conditions:
\begin{enumerate}
\item $\kappa$ is invertible,
\item $\kappa$ is skewon-free,
\item $\kappa$ is \emph{birefringent}, that is, the Fresnel surface of
$\kappa$ is the union of two distinct light cones.
\end{enumerate}
We show that there are only three classes of mediums with these
properties and give explicit expressions in local coordinates for each
class.
\end{abstract}

\maketitle

We will study the \emph{pre-metric} Maxwell's equations.  In this
setting Maxwell's equations are written on a $4$-manifold $N$ and the
electromagnetic medium is described by an antisymmetric $2\choose
2$-tensor $\kappa$ on $N$.
Then the electromagnetic medium $\kappa$ determines a fourth order
polynomial surface in each cotangent space called the \emph{Fresnel
  surface} $\cF$ and it acts as a generalisation of the light cone
determined by a Lorentz metric; the Fresnel surface parameterises
wave-speeds as a function of direction \cite{Rubilar2002, Obu:2003,  PunziEtAl:2009}.
At each point in spacetime $N$, the electromagnetic medium depends, in
general, on 36 free components. In this work we assume that the medium
is \emph{skewon-free}. Then there are only 21 free components and such
medium describe non-dissipative medium.  For example, in skewon-free
medium Poynting's theorem holds under suitable assumptions.

The above means that in the pre-metric setting we have two
descriptions of electromagnetic medium: First, we have the $2\choose
2$-tensor $\kappa$ that contains the coefficients 
in Maxwell's equations.  On the other hand, we also have the Fresnel
surface $\cF$, which describes the behaviour of a wavespeed for a
propagating electromagnetic wave. If $\kappa$ is known we can always
compute $\cF$ by an explicit equation (see equation \eqref{eq:Fr}).
A more challenging question is to understand the converse dependence,
or inverse problem: If the Fresnel surface $\cF\vert_p$ is known for
some $p\in N$, what can we say about $\kappa\vert_p$?  
Essentially this asks that if the behaviour of wave speed 
for an electromagnetic medium is known, what can we say about the
medium?
These questions are of theoretical interest, but also of practical
interest as they relate to understanding measured data in engineering
applications like traveltime tomography in anisotropic medium.
We will here only study the problem at a point $p\in N$ since the
dependence will never be unique. For example, the Fresnel surface
$\cF$ is always invariant under scalings and inversions of
$\kappa$ \cite{Obu:2003, Dahl:2011:Closure}. In general, these are not
the only invariances, and for a general $\kappa$ the relation between
$\kappa$ and $\cF$ does not seem to be very well understood.

A natural first task is to characterise those mediums $\kappa$ for
which the Fresnel surface $\cF$ is the light cone of a Lorentz metric
$g$.  This question was raised in \cite{ObukhovHehl:1999,
  ObuFukRub:00}.  A partial solution was given in \cite{ObuFukRub:00},
and (in skewon-free mediums with real coefficients) the complete solution was given in
\cite{FavaroBergamin:2011}.  The result is that if the Fresnel surface
is a light cone, then $\kappa$ is necessarily proportional to a Hodge
star operator (plus, possibly, an axion component proportional to the
identity). For an alternative proof, see \cite{Dahl:2011:Closure}, and
for related results, see \cite{ObuRub:2002,LamHeh:2004,Itin:2005}.

The main contribution of this paper is Theorem
\ref{thm:factorMedium}. It gives a pointwise characterisation of all
electromagnetic medium tensors $\kappa$ with real coefficients such
that
\begin{enumerate}
\item $\kappa$ is invertible,
\item $\kappa$ is skewon-free,
\item $\kappa$ is \emph{birefringent}, that is, the Fresnel surface of $\kappa$ is the union of two distinct
  light cones for Lorentz metrics.
\end{enumerate}
The first two assumptions imply that $\kappa$ is essentially in one-to-one
correspondence with an \emph{area metric}. Area metrics
also appear when studying the propagation of a photon in
a vacuum with a first order correction from quantum electrodynamics
\cite{DruHath:1980, Schuller:2010}.
The Einstein field equations have also been generalised
into equations where the unknown field is an area metric
\cite{PSW_JHEP:2007}.
For further examples, see \cite{PunziEtAl:2009, Schuller:2010} and
for the differential geometry of area metrics, see 
\cite{SchullerWohlfarth:2006, PSW_JHEP:2007}. 
For Maxwell's equations, the interpretation of condition \emph{(iii)}
is that differently polarised waves can propagate with different wave
speeds.  In such medium one should expect that propagation of
electromagnetic waves is determined by null-geodesics of two
Lorentz metrics. A typical example of such medium is a uniaxial
crystal
For partial results describing when the Fresnel surface factorises,
see \cite{RiveraSchuller:2011} and in 3 dimensions, see \cite{Kachalov:2004, DahlPIER:2006}.

In Theorem
\ref{thm:factorMedium} we show that there are only 
three medium classes with the above  properties and we
give explicit expressions in local coordinates for each class. 
Of these classes, the first is a generalisation of uniaxial medium
and the last seems to be unphysical; heuristic arguments suggest that
Maxwell's equations are not hyperbolic in the last class.

The main idea of the proof is as follows.
We will use the normal form theorem for area metrics
derived in \cite{Schuller:2010}, which pointwise divides area metrics
into $23$ metaclasses and gives explicit expressions in
local coordinates for each metaclass. 
This result was also used in
\cite{FavaroBergamin:2011}, and by \cite{Schuller:2010}
we only need to consider the first 7 metaclasses.
For each of these metaclasses, the Fresnel surface
can
be written as $\cF\vert_p = \{ \xi\in \setR^4 : f(\xi)=0\}$ for a 
homogeneous $4$th order polynomial $f\colon \setR^4\to \setR$ with
coefficients determined by $\kappa\vert_p$. Since
$\kappa\vert_p$ is birefringent, $f$ factorises as
\begin{eqnarray*}
   f(\xi) &=& f_+(\xi) f_-(\xi), \quad \xi\in \setR^4
\end{eqnarray*}
into homogeneous $2$nd order polynomials $f_\pm \colon \setR^4\to
\setR$. By identifying coefficients we obtain a system of polynomial
equations in coefficients of $f$ and $f_\pm$. In the last step we
eliminate the coefficients in $f_\pm$ from these equations whence we
obtain constraints on $f$ (and hence on $\kappa$) that much be
satisfied when $\kappa$ is birefringent. To eliminate variables we use
the technique of \emph{Gr\"obner bases}, which was also used in
\cite{Dahl:2011:Closure}.

A limitation of Theorem \ref{thm:factorMedium} is that the explicit
expression is only valid at a point. The reason for this is that the
decomposition in \cite{Schuller:2010} essentially relies on the Jordan
normal form theorem for matrices, which is unstable under
perturbations. 
Another limitation is that we do not allow for complex coefficients in
$\kappa$. Therefore mediums like \emph{chiral medium} are not included
in the mediums in Theorem \ref{thm:factorMedium}.

This paper relies on computations by computer algebra. For information
about the Mathematica notebooks for these computations, please see the
author's homepage.

\section{Maxwell's equations}
\label{mainSec}
By a \emph{manifold} $M$ we mean a second countable topological Hausdorff
space that is locally homeomorphic to $\setR^n$ with $C^\infty$-smooth
transition maps. All objects are assumed to be smooth and real where defined.  
Let $TM$ and $T^\ast M$ be the tangent and cotangent bundles,
respectively.
For $k\ge 1$, let $\Lambda^k(M)$ be the set of
antisymmetric $k$-covectors, so that $\Lambda^1(N)=T^\ast N$.  
Also, let $\Omega^k_l(M)$ be $k\choose l$-tensors that are
antisymmetric in their $k$ upper indices and $l$ lower indices. In
particular, let $\Omega^k(M)$ be the set of $k$-forms.
 Let $C^\infty(M)$ be the set of functions. 
%
%
The Einstein summing convention is used throughout. When writing tensors 
in local coordinates we assume that the components satisfy the same symmetries as
the tensor. 

\subsection{Maxwell's equations on a $4$-manifold}
\label{sec:MaxOn4}
On a $4$-manifold $N$, Maxwell's equations read 
\begin{eqnarray}
\label{max4A}
dF &=& 0, \\
\label{max4B}
dG &=& j,
\end{eqnarray}
where $d$ is the exterior derivative on $N$, $F,G\in \Omega^2 (N)$ are
the electromagnetic field variables, and $j\in \Omega^3(N)$ is the
source term.
%
By an \emph{electromagnetic medium} on $N$ 
we mean a map 
\begin{eqnarray*}
   \kappa \colon \Omega^2(N) &\to& \Omega^2(N).
\end{eqnarray*} 
We then say that $2$-forms $F,G\in \Omega^2(N)$ \emph{solve Maxwell's
  equations in medium $\kappa$} if  $F$ and $G$ satisfy equations
\eqref{max4A}--\eqref{max4B} and
\begin{eqnarray}
\label{FGchi}
  G &=& \kappa(F).
\end{eqnarray}
Equation \eqref{FGchi} is known as the \emph{constitutive equation}.
If $\kappa$ is invertible, 
it follows that one can eliminate half of the free variables in
Maxwell's equations \eqref{max4A}--\eqref{max4B}.  We assume that
$\kappa$ is linear and determined pointwise so that we can represent
$\kappa$ by an antisymmetric $2\choose 2$-tensor $\kappa \in
\Omega^2_2(N)$. If in coordinates $\{x^i\}_{i=0}^3$ for $N$ we have
\begin{eqnarray}
\label{eq:kappaLocal}
  \kappa &=& \frac 1 2 \kappa^{ij}_{lm} dx^l\otimes dx^m\otimes \pd{}{x^i}\otimes \pd{}{x^j}
\end{eqnarray}
and $F = F_{ij} dx^i \otimes dx^j$ and $G = G_{ij} dx^i \otimes dx^j$, 
then constitutive equation \eqref{FGchi} reads
\begin{eqnarray}
 \label{FGeq_loc}
   G_{ij} &=& \frac 1 2 \kappa_{ij}^{rs} F_{rs}.
\end{eqnarray}
Then at each point on $N$, a general antisymmetric $2\choose 2$-tensor
$\kappa$ depends on $36$ free real components.
In the main result of this paper (Theorem \ref{thm:factorMedium}) we
will assume that $\kappa$ is \emph{skewon-free}, that is,
$\kappa(u)\wedge v= u\wedge \kappa(v)$ for all $u,v\in \Omega^2_2(N)$
whence $\kappa$ has only 21 free components. Physically, such medium
describe non-dissipative medium; if $\kappa$ is time-independent, then
Poynting's theorem holds under suitable assumptions \cite[Proposition
3.3]{Dahl:2009}.  Let us also note that if $N$ is orientable, then
invertible skewon-free mediums are essentially in a one-to-one
correspondence with area metrics. See \cite{Schuller:2010} and
\cite[Proposition 2.4]{Dahl:2011:Conjugation}. The medium is called
\emph{axion-free} if $\operatorname{trace} \kappa = 0$ \cite{Obu:2003}.

%
By a \emph{pseudo-Riemann metric} on a manifold $M$ we mean a
symmetric $0\choose 2$-tensor $g$ that is non-degenerate. If $M$ is
not connected we also assume that $g$ has constant signature. 
By a \emph{Lorentz metric} we mean a pseudo-Riemann metric on a
$4$-manifold with signature $(-+++)$ or $(+---)$.
%
The \emph{light cone} of a Lorentz metric is defined as 
\begin{eqnarray*}
  N(g) &=& \{\xi \in T^\ast_p(N) : g(\xi,\xi) = 0\}.
\end{eqnarray*}
For $p\in N$ we define $N_p(g) = N(g)\cap T^\ast_p(N)$.

If $g$ is a pseudo-Riemann metric on a
orientable $4$-manifold $N$, then the Hodge operator of $g$ induces a skewon-free
$2\choose 2$-tensor that we denote by $\kappa=\ast_g$. Moreover, if locally
$g=g_{ij}dx^i\otimes dx^j$, and $\kappa$ is
written as in equation \eqref{eq:kappaLocal}, then
\begin{eqnarray}
\label{eq:hodgeKappaLocal}
\kappa^{ij}_{rs} &=& \sqrt{\vert g\vert} g^{ia}g^{jb} \varepsilon_{abrs},
\end{eqnarray}
where $\det g= \det g_{ij} $, $g^{ij}$ is the
$ij$th entry of $(g_{ij})^{-1}$, and $\varepsilon_{l_1\cdots l_n}$ is
the \emph{Levi-Civita permutation symbol}. We treat
$\varepsilon_{l_1\cdots l_n}$ as a purely combinatorial object (and
not as a tensor density). Let also $\varepsilon^{l_1\cdots l_4}=
\varepsilon_{l_1\cdots l_4}$.

\subsection{Representing $\kappa$ as a $6\times 6$ matrix}
\label{sec:Rep6x6}
Let $O$ be the ordered set of index pairs $\{ 01, 02, 03$, $23, 31, 12\}$. 
If $I\in O$, we  denote the individual indices by $I_1$ and $I_2$. Say,
if $I=31$ then $I_2=1$. 

If $\{x^i\}_{i=0}^3$ are local coordinates for a $4$-manifold $N$, and
$J\in O$ we define $dx^J = dx^{J_1}\wedge dx^{J_2}$.  
A basis for $\Omega^2(N)$
is given by $\{ dx^{J}: J\in O\}$, that is,
\begin{eqnarray}
\label{eq:2basis}
\{ dx^0\wedge dx^1, dx^0\wedge dx^2, dx^0\wedge dx^3, dx^2\wedge dx^3, dx^3\wedge dx^1, dx^1\wedge dx^2\}.
\end{eqnarray}
This choice of basis follows \cite[Section A.1.10]{Obu:2003} and
\cite{FavaroBergamin:2011}.
If $\kappa\in \Omega^2_2(N)$ is written as in equation
\eqref{eq:kappaLocal} and $J\in O$, then
\begin{eqnarray}
\label{eq:kappaMatDef}
\kappa(dx^{J}) = \sum_{I\in O} \kappa^J_I dx^I, \quad J \in O,
\end{eqnarray}
where $\kappa^J_I = \kappa^{J_1 J_2}_{I_1 I_2}$. We will always use
capital letters $I,J,K,\ldots$ to denote elements in $O$.  Let $b$ be
the natural bijection $b\colon O\to \{1,\ldots, 6\}$.  Then we
identify coefficients $\{\kappa^J_I : I,J\in O\}$ for $\kappa$ with the $6\times 6$
matrix $A=(\kappa^J_I)_{IJ}$ defined as $\kappa^J_I = A_{b(I) b(J)}$
for $I,J\in O$ \cite{Dahl:2011:Conjugation}.

\subsection{Fresnel surface}
Let 
$\kappa\in \Omega^2_2(N)$ on a $4$-manifold $N$. 
If $\kappa$ is locally given by equation \eqref{eq:kappaLocal} in coordinates $\{x^i\}_{i=0}^3$, let 
\begin{eqnarray*}
\cG^{ijkl}_0 &=& \frac 1 {48} 
\kappa^{a_1 a_2}_{b_1 b_2} 
\kappa^{a_3 i}_{b_3 b_4} 
\kappa^{a_4 j}_{b_5 b_6} 
\varepsilon^{b_1 b_2 b_5 k} 
\varepsilon^{b_3 b_4 b_6 l} 
\varepsilon_{a_1 a_2 a_3 a_4}.
\end{eqnarray*}
In overlapping coordinates $\{\widetilde x^i\}_{i=0}^3$, these coefficients
transform as
\begin{eqnarray}
\label{eq:TRtrans}
\widetilde \cG_0^{ijkl} &=& \det \left(\pd{x^r}{\widetilde x^s}\right)\, \cG_0^{abcd} \pd{\widetilde x^i}{x^a} \pd{\widetilde x^j}{x^b}\pd{\widetilde x^k}{x^c}\pd{\widetilde x^l}{x^d},
\end{eqnarray}
and components $\cG^{ijkl}_0$ define a tensor density $\cG_0$ on $N$
of weight $1$. The \emph{Tamm-Rubilar tensor density} is the symmetric
part of $\cG_0$ and we denote this tensor density by $\cG$
\cite{Rubilar2002, Obu:2003, Dahl:2011:Closure}. In coordinates,
$\cG^{ijkl} = \cG^{(ijkl)}_0$, where parenthesis indicate that indices
$ijkl$ are symmetrised with scaling $1/4!$.  Using tensor density
$\cG$, the \emph{Fresnel surface} at a point $p\in N$ is defined as
\begin{eqnarray}
\label{eq:Fr}
\cF\vert_p &=& \{\xi\in \Lambda^1_p(N) : \cG^{ijkl} \xi_i \xi_j \xi_k \xi_l  = 0\}.
\end{eqnarray}
The Fresnel surface is a fundamental object when studying wave
propagation in Maxwell's equations.  
It is clear that in each cotangent space, the Fresnel surface
$\cF\vert_p$ is a fourth order polynomial surface, so it can have 
multiple sheets and singular points. 
There are various ways to derive
the Fresnel surface; by studying a propagating weak singularity
\cite{ObuFukRub:00, Rubilar2002, Obu:2003}, using a geometric optics  \cite{Itin:2009,
  Dahl:2011:Closure}, or as the characteristic polynomial of the full Maxwell's equations
\cite{Schuller:2010}. Classically, the Fresnel surface can be seen as
the dispersion equation for a medium, so that it constrains possible
wave speed(s) as a function of direction. 

If 
$\kappa = f\ast_g$
for a Lorentz metric $g$ and a non-zero function $f\in C^\infty(N)$,
then the Fresnel surface is the light cone of $g$. The converse is also
true (assuming that $\kappa$ is skewon-free and axion-free)
\cite{Obu:2003, FavaroBergamin:2011, Dahl:2011:Closure}. The medium
given by $\kappa = f\ast_g$ is known as \emph{non-birefringent
  medium}. For such medium propagation speed does not depend on
polarisation.

 
\begin{definition}
\label{def:bire}
Suppose $N$ is a $4$-manifold, $\kappa\in \Omega^2_2(N)$ and
$\mathscr{F}$ is the Fresnel surface for $\kappa$. If $p\in N$ we say
that $\mathscr{F}\vert_p$ \emph{decomposes into a double light cone} if
there exists Lorentz metrics $g_+$ and $g_-$ defined in a
neighbourhood of $p$ such that
\begin{eqnarray}
\label{eq:cFdecomp}
\cF \vert_p &=& N_p(g_+)\,\, \cup \,\, N_p(g_-)
\end{eqnarray}
and $N_p(g_+)\neq N_p(g_-)$. We then also say that $\kappa\vert_P$ is \emph{birefringent}.
\end{definition}

We know that two Lorentz metrics are conformally related if and only
if they have the same light cones \cite{Ehrlich:1991}. The condition
$N_p(g_+)\neq N_p(g_-)$ thus only exclude non-birefringent mediums,
which for skewon-free mediums are well understood (see above).
When $\cF\vert_p$ decompose into a double light cone as in Definition
\ref{def:bire} a physical interpretation is that the medium is
\emph{birefringent}. That is, differently polarised electromagnetic
waves can propagate with different wave speeds.  Common examples of
such mediums are \emph{uniaxial crystals} like calcite
 \cite[Section 15.3]{BornWolf}.

To prove of the next four propositions we will need some terminology
from algebraic geometry. If $k=\setR$ or $k=\setC$, we denote by
$k[x_1, \ldots, x_n]$ the ring of polynomials $k^n\to k$ in variables
$x_1, \ldots, x_n$. Moreover, a non-constant polynomial $f\in k[x_1,
\ldots, x_n]$ is \emph{irreducible} if $f=uv$ for $u,v \in k[x_1,
\ldots, x_n]$ implies that $u$ or $v$ is a constant.  For a polynomial
$r\in k[x_1, \ldots, x_n]$, let $V(r)=\{x\in k^n : r(x)=0\}$ be the
\emph{variety} induced by $r$, and let $\langle r \rangle=\{ fr : f\in
k[x_1, \ldots, x_n]\}$ be the the \emph{ideal generated by $r$}.  For
what follows the necessary theory for manipulating these objects can,
for example, be found in \cite{CoxLittleOShea:1992}.

The next proposition can be seen as a reformulation of the \emph{Brill
  equations} that characterise when a homogeneous polynomial factorises
into linear forms \cite{Briand:2010}.

\begin{proposition} 
\label{lemma:irr}
Suppose $Q\in \setC^{4\times 4}$ is a symmetric non-zero matrix and
$f\in \setC[\xi_0, \ldots, \xi_3]$ is the polynomial $f(\xi) =
\xi^t\cdot Q\cdot \xi$ for $\xi=(\xi_0, \ldots, \xi_3) \in \setC^4$.
Then $f$ is irreducible in $\setC[\xi_0, \ldots, \xi_3]$ if and only
if $\operatorname{adj} Q \neq 0$.
\end{proposition}

In Proposition \ref{lemma:irr}, $\operatorname{adj} Q$ is the
\emph{adjugate matrix} of all cofactor expansions of $Q$,
and $\xi^t$ is the matrix transpose. 

\begin{proof}
  The result follows from the following three facts: First, if $f =
  uv$ for polynomials $u,v\in \setC[\xi_0, \ldots, \xi_3]$ then $u$
  and $v$ are linear.  (To see this, we know that $Q(0)=0$, so we may
  assume that $u(0)=0$.  For a contradiction, suppose that $v(0)\neq
  0$. Then $df\vert_0=0$ implies that $du\vert_0=0$, but then $u=0$,
  so $Q=0$.)  Second, by \cite[Example 2] {Briand:2010} polynomial $f$
  is a product of linear forms or $f$ has a multiple factor if and
  only if $G_f(\xi, \eta, \zeta)=0$ for all $\xi,\eta,\zeta\in
  \setC^4$. Here $G_f$ is the \emph{Gaeta covariant} defined as
\begin{eqnarray*}
  G_f(\xi, \eta, \zeta) &=& -\frac 1 2  \det \begin{pmatrix}
  2f(\xi) &  (Df)_\xi(\eta) &  (Df)_\xi(\zeta) \\
   (Df)_\xi(\eta) & 2f(\eta) &  (Df)_\eta(\zeta) \\
   (Df)_\xi(\zeta) & (Df)_\eta(\zeta) & 2f(\zeta) \\
\end{pmatrix},
\end{eqnarray*}
and $(Df)_a(b)=\frac {d}{dt}(f(a+tb))\vert_{t=0}$ is the directional
derivative.  Third, by computer algebra we have $\operatorname{adj} Q
=0$ if and only if $G_f=0$.
\end{proof}

In \cite{montaldi:2006} it is proven that the light cone of a Lorentz
metric can not contain a vector subspace of dimension $\ge 2$. The
next proposition generalise this result to double light cones. In the
proof of Theorem \ref{thm:factorMedium} we will use this result to
show that medium tensors in the last 16 metaclasses in the
classification of \cite{Schuller:2010} can not decompose into a double
light cone. In \cite{Schuller:2010} this property was used to show
that these last metaclasses are neither hyperbolic, so in these
metaclasses, Maxwell's equations are not well-posed.

\begin{proposition}
\label{prop:noTwoDimSubspace} 
Suppose $g_\pm$ are Lorentz metrics on a $4$-manifold $N$.  If
$\Gamma\subset T_pN$ is a non-empty vector subspace such that
$\Gamma\subset N(g_+)\cup N(g_-)$, then $\operatorname{dim} \Gamma\le
1$.
\end{proposition}

\begin{proof}
We may assume that $\dim \Gamma\ge 1$.
  Let us first prove the result in the special case that $g_+$ and
  $g_-$ are conformally related (after \cite[Proposition
  2]{montaldi:2006}).  Let $\{x^i\}_{i=0}^3$ be coordinates around $p$
  such that $g_+\vert_p = \pm \operatorname{diag}(-1,1,1,1)$, whence
  we can identify $T_pM$ and $\setR \oplus \setR^3$.  Let $\pi$ be the
  Cartesian projection onto the first component.  Then the restriction
  $\pi\vert_\Gamma\colon \Gamma\to \setR$ satisfies
  $\operatorname{ker} \pi\vert_\Gamma=\{0\}$ and $\dim
  \operatorname{range} \pi\vert_\Gamma=1$, and the claim follows.

  For general $g_+$ and $g_-$ let us show that $\dim \Gamma\ge 2$ 
leads to a 
contradiction. If $\dim \Gamma\ge 2$, we can find linearly independent
  $u,v\in \Gamma$ such that  $\operatorname{span}\{u,v\}\subset N(g_+)\cup N(g_-)$.
  We may further assume that $u\in N(g_+)$. Let
\begin{eqnarray*}
    U &=& \{ \theta\in \setR : \cos\theta u + \sin \theta v \not\in N(g_-) \}.
\end{eqnarray*}
For $w\in T^\ast N$ let us write $\Vert w\Vert^2 = g_+(w,w)$.  
If $U$ is
empty, then $\operatorname{span}\{u,v\} \subset N(g_-)$ and the result
follows from the special case.  Otherwise there exists a $\theta_0\in
U$ so that $ \Vert \cos\theta u + \sin \theta v \Vert^2 =0$ for all
$\theta$ in some neighbourhood $I_0\ni \theta_0$. Differentiating
gives
\begin{eqnarray*}
\frac 1 2 \left( \Vert v\Vert^2 - \Vert u\Vert^2 \right)\,\cdot \,\sin 2\theta 
+  g_+(u,v)\, \cdot \, \cos 2\theta  &=& 0, \quad \theta\in I_0.
\end{eqnarray*}
By computing the Wronskian of $\sin 2\theta, \cos 2\theta$, it follows
that $0=\Vert u\Vert^2 = \Vert v\Vert^2$ and
$g_+(u,v)=0$. Thus $\operatorname{span}\{u,v\}\subset N(g_+)$, 
but this contradicts the special case.
\end{proof}

The next proposition gives the pointwise description of the
Tamm-Rubilar tensor density at points $p\in N$ where the the Fresnel
surface decomposes into a double light cone. Let us emphasize that the
result is pointwise. For example, in equation \eqref{prop:GdecompIntoGpm} the two
sides have different transformation rules.

\begin{proposition}
\label{thm:localDescGc}
Suppose $N$ is a $4$-manifold, $\kappa\in \Omega^2_2(N)$, and the
Fresnel surface of $\kappa$ decomposes into a double light cone at
$p\in N$. If $\{x^i\}_{i=0}^3$ are coordinates around $p$ and
$\cG^{ijkl}$ 
and $g_\pm = g_{\pm\, ij} dx^i\otimes dx^j$ are as in Definition
\ref{def:bire}, then
\begin{eqnarray}
\label{prop:GdecompIntoGpm}
  \cG^{ijkl} \xi_i \xi_j \xi_k \xi_l &=& C\,\, (g_+^{ij} \xi_i \xi_j)\,\,(g_-^{kl} \xi_k \xi_l)\,\, \mbox{ at} \,\,p, \quad \{\xi_i\}_{i=0}^3\in \setR^4,
\end{eqnarray}
for some $C\in \setR\slaz$.
\end{proposition}

\begin{proof}
  Let $f_\pm, \gamma\colon \setC^4\to \setC$ be polynomials
  $f_\pm(\xi) = g_{\pm}^{ij} \xi_i \xi_j$ and $\gamma(\xi)= \cG^{ijkl}
  \xi_i \xi_j \xi_k \xi_l$ for $\xi=(\xi_{i})_{i=0}^3\in \setC^4$.
For ideals $I_\pm = \langle f_\pm \rangle$ we then have $\langle f_+
f_-\rangle = I_+ \cap I_-$ whence equation \eqref{eq:cFdecomp} implies
that $V(\langle \gamma\rangle )
= V(I_+\cap I_-)$ and
passing to ideals gives
\begin{eqnarray*}
I(V (\langle \gamma \rangle)) &=& I(V(I_+\cap I_-)).
\end{eqnarray*}
The Strong Nullstellensatz implies that \cite[p. 175]{CoxLittleOShea:1992}
\begin{eqnarray*}
\langle \gamma \rangle &\subset &  \sqrt{I_+}\cap \sqrt{I_-},
\end{eqnarray*}
where $\sqrt{I}$ be the \emph{radical} of an ideal $I$ and we used
identities $\sqrt{I\cap J} = \sqrt{I}\cap \sqrt{J}$ and $I\subset
\sqrt{I}$ valid for any ideals $I$ and $J$
 \cite[Proposition 16 in Section 4.3]{CoxLittleOShea:1992}. 
By Proposition \ref{lemma:irr}, $f_\pm$ are irreducible polynomials, 
so $I_\pm$ are prime ideals and $I_\pm = \sqrt{I_\pm}$. Thus
\begin{eqnarray*}
\langle \gamma \rangle &\subset &  \langle f_+ f_-\rangle
\end{eqnarray*}
and  $\gamma =  p f_+  f_-$  for some  polynomial $p\colon  \setC^4\to
\setC$.  Computing  the polynomial degree  of both sides  implies that
$p$   is  a   non-zero  constant   $p=C\in  \setC$.    By  Proposition
\ref{prop:noTwoDimSubspace}, we  can find a $\xi \in  \setR^4$ so that
$\xi\not\in V(f_+)\cup V(f_-) = V(\gamma)$ whence $C\in \setR\slaz$.
\end{proof}

For two Lorentz metrics $g$ and $h$ we know that their light cones
$N(g)$ and $N(h)$ coincide if and only if $g$ and $h$ are conformally
related \cite{Ehrlich:1991}. The next proposition gives an analogous uniqueness
result for Fresnel surfaces that decompose into a double light
cone. 

\begin{proposition}
\label{thm:DoubleObservability}
Suppose $N$ is a $4$-manifold, $\kappa\in \Omega^2_2(N)$ and the
Fresnel surface of $\kappa$ decomposes into a double light cone at
$p\in N$. Suppose furthermore that equation \eqref{eq:cFdecomp} holds
for Lorentz metrics $g_\pm$ and for Lorentz metrics $h_\pm$.  Then
there exists constants $C_\pm \in \setR\slaz$ such that exactly one of
the following conditions hold:
$g_\pm =C_\pm \, h_\pm$ at $p$ or
$g_\mp =C_\pm \, h_\pm$ at $p$.
\end{proposition}

\begin{proof} 
  The result follows by Propositions \ref{lemma:irr} and
  \ref{thm:localDescGc} and since any polynomial has a unique
  decomposition into irreducible factors \cite[Theorem 5 in Section
  3.5]{CoxLittleOShea:1992}.
\end{proof}

The next example shows that unique decomposition is not true when
$\cF$ only decompose into second order surfaces.

\begin{example}
In
coordinates $\{x^i\}_{i=0}^3$ for $\setR^4$, let $\kappa$ be the skewon-free 
$2\choose 2$-tensor determined by the 
$6\times 6$ matrix
  $(\kappa_{I}^J)_{IJ} =\operatorname{diag}(-1,1,0,-1,1,0)$.
Then $\kappa$ has Fresnel surface
\begin{eqnarray*}
 \cF &=& \{ \xi\in T^\ast \setR^4 : \xi_0 \xi_1 \xi_2 \xi_3 = 0\}.
\end{eqnarray*}
It is clear that $\cF$ has multiple factorisations into quadratic
forms, and by Proposition \ref{prop:noTwoDimSubspace}, $\cF$ does
not factorise into a double light cone.  \proofBox
\end{example}

\newcommand{\rr}[0]{\alpha}
\newcommand{\pos}[0]{\beta}
\newcommand{\sqr}[0]{\sqrt{2}}

\section{Non-dissipative media with a double light cone }

Theorem \ref{thm:factorMedium} below is the main result of this
paper. To formulate the theorem we first need some terminology.  Suppose $L \colon V\to V$ is a linear map where $V$ is a
$n$-dimensional real vector space.  If the matrix representation of
$L$ in some basis is $A\in \setR^{n\times n}$ and $A$ is written using
the Jordan normal form
we say that
$L$ has \emph{Segre type} $\left[m_1\cdots m_r\,
  k_{1}\overline{k_{1}}\cdots k_{s}\overline{k_{s}}\right]$ when the
blocks corresponding to real eigenvalues have dimensions $m_1\le
\cdots \le m_r$ and the blocks corresponding to complex eigenvalues
have dimensions $2k_1\le \cdots \le 2k_s$.
Moreover, by uniqueness of the Jordan normal form, the Segre type
depends only on $L$ and not on the basis. 
For a $2\choose 2$-tensor $\kappa$ on a $4$-manifold, we define the
Segre type of $\kappa\vert_p$ as the Segre type of the linear map
$\Omega^2(N)\vert_p\to \Omega^2(N)\vert_p$ in basis \eqref{eq:2basis}.
By counting how many ways a $6\times 6$ matrix can be
decomposed into Jordan normal forms, it follows that there are only
$23$ Segre types for a $2\choose 2$-tensor. A main result of
\cite{Schuller:2010} are simple normal forms in local coordinates for
each of these Segre types under the assumption that $\kappa\vert_p$ is
skewon-free and invertible. In the below proof we will use the
restatement of this result in \cite{Dahl:2011:Conjugation}.

In the proof of Theorem \ref{thm:factorMedium} we will eliminate
variables in systems of polynomial equations. 
%
Suppose $V\subset \setC^n$ is the solution set to polynomial equations
$f_1=0, \ldots, f_N=0$ where $f_i\in \setC[x_1, \ldots, x_n]$.  
If $I$ is
the ideal generated by $f_1, \ldots, f_N$, the \emph{elimination
  ideals} are the polynomial ideals defined as
\begin{eqnarray*}
\label{eq:elimId}
  I_k&=& I\cap \setC[x_{k+1}, \ldots, x_{n}],\quad  k\in \{0,\ldots, n-1\}.
\end{eqnarray*} 
Thus, if $(x_1, \ldots, x_n)\in V$ then $p(x_{k+1}, \ldots, x_n)=0$
for any $p\in I_k$, and $I_k$ contain polynomial consequences of the
original equations that only depend on variables $x_{k+1}, \ldots,
x_n$.  Using Gr\"obner basis, one can explicitly compute $I_k$
\cite[Theorem 2 in Section 3.1]{CoxLittleOShea:1992}.  In the proof
of Theorem \ref{thm:factorMedium}, we will use the built-in
Mathematica routine \textsf{GroebnerBasis} for such computations.

\begin{theorem}
\label{thm:factorMedium}
Suppose $N$ is a $4$-manifold and  $\kappa\in \Omega^2_2(N)$. Furthermore, 
suppose that at some $p\in N$ 
\begin{enumerate}
\item $\kappa\vert_p$ has no skewon component, 
\item $\kappa\vert_p$ is invertible as a linear map $\Lambda^2_p(N)\to \Lambda^2_p(N)$,
\item the Fresnel surface $\cF\vert_p$ factorises into a double light cone.
\end{enumerate}
Then $\kappa\vert_p$ must have Segre type 
  $[1\overline{1}\,1\overline{1}\,1\overline{1}]$, 
$[2\overline{2}\,1\overline{1}]$ or
$[11\,1\overline{1}\,1\overline{1}]$.
\begin{enumerate}
\item \textbf{Metaclass I:} If $\kappa\vert_p$ has Segre type
  $[1\overline{1}\,1\overline{1}\,1\overline{1}]$, there are
  coordinates $\{x^i\}_{i=0}^3$ around $p$ such that 
\begin{eqnarray*}
(\kappa_I^J)_{IJ} &=& 
\begin{pmatrix}
\alpha_1 & 0 & 0 & -\pos_1 & 0 &0 \\
0 & \alpha_2 & 0 & 0  &-\pos_2 & 0 \\
0 & 0 &  \alpha_3 & 0&0 &-\pos_3  \\
\pos_1 & 0 &  0 & \alpha_1 & 0 &0  \\
0 & \pos_2 & 0 & 0  &\alpha_2 & 0 \\
0 & 0 & \pos_3 & 0 & 0 &\alpha_3 
\end{pmatrix}.
\end{eqnarray*}
for some for some $\alpha_1, \alpha_2, \alpha_3 \in \setR$ and
$\beta_1, \beta_2,\beta_3>0$.  For $i\in \{1,2,3\}$ let
\begin{eqnarray}
\label{eq:MetaclassI-D-def}
 D_{i} &=& \frac{ (\alpha_{i'}-\alpha_{i''})^2+\beta_{i'}^2+\beta_{i''}^2}{\beta_{i'} \beta_{i''}}
\end{eqnarray}
where $i'$ and $i''$ are defined such that $\{i,i',i''\}=\{1,2,3\}$ and $i'<i''$.

Then, for exactly one $i\in \{1,2,3\}$ we have
\begin{eqnarray}
\label{eq:mc-I-Dsol}
   D_i=2,  \quad D_{i'} = D_{i''}
\end{eqnarray}
and equation \eqref{eq:cFdecomp} holds for Lorentz metrics 
\begin{eqnarray*}
  g_\pm & =& \operatorname{diag}\left(1, 
\frac 1 2 \left( -D_{1} \pm \sqrt{D_{1}^2-4}\right),
\frac 1 2 \left( -D_{2} \pm \sqrt{D_{2}^2-4}\right),
\frac 1 2 \left( -D_{3} \pm \sqrt{D_{3}^2-4}\right)\right)^{-1}.
\end{eqnarray*}

\item \textbf{Metaclass II:}
If $\kappa\vert_p$ has Segre type $[2\overline{2}\,1\overline{1}]$, 
If $\kappa\vert_p$ is in Metaclass II, there are coordinates $\{x^i\}_{i=0}^3$ around $p$ such that 
\begin{eqnarray*}
(\kappa_I^J)_{IJ} &=& 
\begin{pmatrix}
\rr_1         &       -\pos_1       &    0 & 0 & 0 & 0 \\
\pos_1      &    \rr_1           &    0&  0  & 0 & 0 \\
0              &                0    &  \rr_2 & 0 &0 &-\pos_2 \\
0              &                1    & 0 & \rr_1 & \pos_1 &0  \\
1              &                0    & 0 & -\pos_1  &\rr_1 & 0 \\
0              &                0    & \pos_2 & 0 & 0 &\rr_2
\end{pmatrix}
\end{eqnarray*}
where $\rr_1,\rr_2\in \setR$ and $\beta_1>0$. Then $\alpha_1=\alpha_2$
and $\beta_1=\beta_2$, and equation \eqref{eq:cFdecomp} holds for Lorentz metrics 
\begin{eqnarray*}
g_\pm = \left(
\begin{array}{cccc}
\pm 1 & 0 & 0 & \beta_1  \\
0 & -\beta_1  & 0 & 0 \\
0 & 0 & -\beta_1  & 0 \\
\beta_1  & 0 & 0 & 0
\end{array}
\right)^{-1}.
\end{eqnarray*}

\item 
\textbf{Metaclass IV:}
If $\kappa\vert_p$ has Segre type $[11\,1\overline{1}\,1\overline{1}]$, 
there are coordinates $\{x^i\}_{i=0}^3$ around $p$ such that  
\begin{eqnarray*}
(\kappa_I^J)_{IJ} &=& 
\begin{pmatrix}
\rr_1 & 0 & 0 & -\pos_1 & 0 &0 \\
0 & \rr_2 & 0 & 0  &-\pos_2 & 0 \\
0 & 0 &  \rr_3 & 0 & 0 &\rr_4  \\
\pos_1 & 0 &  0 & \rr_1 & 0 &0  \\
0 & \pos_2 & 0 & 0  &\rr_2 & 0 \\
0 & 0 & \rr_4 & 0 & 0 &\rr_3
\end{pmatrix}
\end{eqnarray*}
for some $\alpha_1, \alpha_2, \alpha_3, \alpha_4 \in \setR$ and $\beta_1, \beta_2>0$. 
Then $\alpha_1 = \alpha_2$, $\beta_1 = \beta_2$, $\alpha_4\neq 0$,
$\alpha_3^2\neq \alpha_4^2$
and 
equation \eqref{eq:cFdecomp} holds for Lorentz metrics 
\begin{eqnarray*}
  g_\pm & =& \operatorname{diag}\left(
1, 
\frac 1 2 \left( -D_{1} \pm \sqrt{D_{1}^2+4}\right),
\frac 1 2 \left( -D_{1} \pm \sqrt{D_{1}^2+4}\right),
-1\right)^{-1},
\end{eqnarray*}
where
\begin{eqnarray}
\label{eq:D1defIV}
   D_{1} &=& \frac{ (\alpha_{2}-\alpha_{3})^2+\beta_{2}^2-\alpha_4^2}{\beta_{2}\alpha_4}.
\end{eqnarray}
\end{enumerate} 
\end{theorem}

\begin{proof}
\textbf{Metaclass I.} 
The local expression for $\kappa\vert_p$ is given by \cite[Theorem 3.2]{Dahl:2011:Closure}.
 Then the Tamm-Rubilar tensor density for $\kappa\vert_p$ satisfies
\begin{eqnarray}
\label{eq:fresnel:metaI}
 C^{-1} \, \cG^{ijkl} \xi_i \xi_j \xi_k \xi_l 
 &=& \xi_0^4 + \xi_1^4 + \xi_2^4 + \xi_3^4 
- D_0 \xi_0 \xi_1 \xi_2 \xi_3 \\
& & \quad\quad\quad\quad\quad
\nonumber
+\sum_{i=1}^3 D_i (\xi_{i'}^2 \xi_{i''}^2 - \xi_0^2 \xi_i^2), \quad\quad\xi\in \setR^4,
\end{eqnarray}
where $C= {\beta_1\beta_2\beta_3}$ and $D_0$ is given explicitly in terms of $\alpha_1, \ldots, \beta_3$,
and implicitly $D_0$ satisfies 
\begin{eqnarray}
\label{eq:MetaclassI_D0}
 D_0^2 = 4\left( 4+ D_1 D_2 D_3 - D_1^2- D_2^2- D_3^2\right).
\end{eqnarray}
By Proposition \ref{thm:localDescGc}, there are real symmetric
matrices $A=(A^{ij})_{i,j=0}^3$ and $B=(B^{ij})_{i,j=0}^3$ such that
\begin{eqnarray} 
\label{eq:metaclassI-fDecomp}
C^{-1} \, \cG^{ijkl} \xi_i \xi_j \xi_k \xi_l  &=& \left( \xi^t\cdot A\cdot \xi\right)\,\,\, \left( \xi^t\cdot B\cdot \xi\right), \quad \xi\in \setR^4.
\end{eqnarray}
Writing out these equations shows that $A^{00} B^{00}=1$. Hence
$A^{00}$ is non-zero, and by rescaling $A$ and $B$, we may assume that
$A^{00}=1$. This substitution simplifies the equations so that by
polynomial substitutions we may eliminate all variables in $B$ and
variable $D_0$. This results in a system of polynomial equations that
only involve $D_1, D_2, D_3$ and the variables in $A$. Further
eliminating the variables in $A$ using a Gr\"obner basis, gives
constraints on $D_1, D_2, D_3$.  By equation
\eqref{eq:MetaclassI-D-def} we know that $D_1, D_2, D_3\ge 2$, whence
these constraints imply that there exists a unique $i\in \{1,2,3\}$
such that conditions \eqref{eq:mc-I-Dsol} hold.  (To see that $i$ is
unique it suffices to note that if $D_i=D_j=2$ for $i\neq j$ then $D_1
= D_2 = D_3 = 2$ whence $D_0=0$ and equation
\eqref{eq:metaclassI-fDecomp} holds for the single light cone
$A=B=g_0^{-1}$ with $g_0 = \operatorname{diag}\{-1,1,1,1\}$. By
Proposition \ref{thm:localDescGc} and unique decomposition into
irreducible factors this gives a contradiction.)  Equations
\eqref{eq:mc-I-Dsol} and \eqref{eq:MetaclassI_D0} imply that $D_0=0$.
The result follows since equation
\eqref{eq:metaclassI-fDecomp} holds with $A=g_+^{-1}$ and $B=g_-^{-1}$,
where $g_\pm$ are the matrices in the theorem formulation.
 
\textbf{Metaclass II.} As in Metaclass I, there are matrices
$A$ and $B$ such that the Fresnel polynomial satisfies equation
\eqref{eq:metaclassI-fDecomp} (with $C = 1$). As in Metaclass I we can
eliminate variables in $B$. Further eliminating all variables in $A$
by a Gr\"obner basis implies that $\alpha_1=\alpha_2$ and $\beta_1 =
\beta_2$. Then a a direct computation shows that equation
\eqref{eq:metaclassI-fDecomp} holds with $A=g_+^{-1}$, $B=
g_-^{-1}$ and $C=\beta_1$.  Computer algebra shows that $g_\pm$ have Lorentz
signatures.

\textbf{Metaclass III.} 
If $\kappa\vert_p$ is in Metaclass III, 
there are coordinates $\{x^i\}_{i=0}^3$ around $p$ such that 
\begin{eqnarray*}
(\kappa_I^J)_{IJ} &=& 
\begin{pmatrix}
\rr_1          &           -\pos_1   &    0   & 0 & 0 & 0 \\
\pos_1         &            \rr_1    &    0   &  0  & 0 & 0 \\
1              &                0    &  \rr_1 & 0 &0 &-\pos_1 \\
0              &                0    &    0   & \rr_1   & \pos_1 &1  \\
0              &                0    &    1   & -\pos_1 &\rr_1 & 0 \\
0              &                1    & \pos_1 &   0     & 0 &\rr_1
\end{pmatrix}
\end{eqnarray*}
where $\rr_1\in \setR$ and $\beta_1>0$. Decomposing the Fresnel
polynomial as in equation \eqref{eq:metaclassI-fDecomp} (with $C =
1$) gives a system of polynomial equations for the variables in $A$,
$B$ and $\kappa\vert_p$. Computing the Gr\"obner basis for these equations
implies that $\beta_1=0$. Thus $\kappa\vert_p$ can not be in Metaclass III.

\textbf{Metaclass IV.}
Let us first note that $\alpha_4 \neq 0$ since otherwise
$\operatorname{span} \{ dx^1\vert_p, dx^2\vert_p\}\subset \cF\vert_p$ which
is not possible by Proposition \ref{prop:noTwoDimSubspace}. 
Then the Tamm-Rubilar tensor density satisfies
\begin{eqnarray*}
C^{-1} \, \cG^{ijkl} \xi_i \xi_j \xi_k \xi_l 
 &=& \xi_0^4-\xi_1^4-\xi_2^4+\xi_3^4 + D_0 \xi_0 \xi_1 \xi_2 \xi_3 \\
 & & \quad \quad + D_{1} (\xi_2^2 \xi_3^2- \xi_0^2 \xi_1^2)
+ D_{2} (\xi_1^2 \xi_3^2- \xi_0^2 \xi_2^2)
+ D_{3} (-\xi_1^2 \xi_2^2- \xi_0^2 \xi_3^2), 
\end{eqnarray*}
where $C = \beta_1 \beta_2 \alpha_4$, $D_0$ is determined
explicitly in terms of $\alpha_1, \ldots, \beta_2$, 
$D_1$ is defined in equation \eqref{eq:D1defIV},  
$D_{3}\ge 2$ is defined in equation \eqref{eq:MetaclassI-D-def} and
\begin{eqnarray*}
   D_{2} &=& \frac{ (\alpha_{1}-\alpha_{3})^2+\beta_{1}^2-\alpha_4^2}{\beta_{1} \alpha_4 }.
\end{eqnarray*}
By decomposing and eliminating variables as in Metaclass I, it follows
that that $D_{0}=0$ and $D_3 = 2$. Thus we have proven that $\alpha_1
= \alpha_2$ and $\beta_1=\beta_2$ whence $D_1 = D_2$ and equation
\eqref{eq:metaclassI-fDecomp} holds with $A=g_+^{-1}$, $B=g_-^{-1}$
and $C$ as above. Moreover, $g_\pm$ both have Lorentz signatures.
Condition $\alpha_3^2\neq \alpha_4^2$ follows since $\det
\kappa\vert_p \neq 0$.

\textbf{Metaclass V.}
If $\kappa\vert_p$ is in Metaclass V, 
there are coordinates $\{x^i\}_{i=0}^3$ around $p$ such that 
\begin{eqnarray*}
(\kappa_I^J)_{IJ} &=& 
\begin{pmatrix}
\rr_1 &       -\pos_1           &           0 & 0 & 0 & 0 \\
\pos_1      &    \rr_1           &           0&  0  & 0 & 0 \\ 
0              &                0     &     \rr_2 & 0 &0 &\alpha_3 \\ 
0              &                1    &           0 & \rr_1 & \pos_1 &0  \\ 
1              &                0    &          0 & -\pos_1  &\rr_1 & 0 \\ 
0              &                0    &  \alpha_3 & 0 & 0 &\rr_2 
\end{pmatrix}
\end{eqnarray*}
where $\alpha_1,\alpha_2, \alpha_3\in \setR$ and $\beta_1>0$.  We may assume that
$\alpha_3\not= 0$, since otherwise  $\operatorname{span}
\{dx^i\vert_p\}_{i=1}^3\subset \cF\vert_p$. Decomposing and
eliminating variables as in Metaclass I gives that
$\beta_1$ is purely complex. Thus $\kappa\vert_p$ can not be in Metaclass V.

\textbf{Metaclass VI.}
If $\kappa\vert_p$ is in Metaclass VI, there are coordinates $\{x^i\}_{i=0}^3$
around $p$ such that
\begin{eqnarray*}
(\kappa_I^J)_{IJ} &=& 
\begin{pmatrix}
\rr_1 & 0 & 0 & -\pos_1 & 0 &0 \\
0 & \rr_2 & 0 & 0  &\rr_4 & 0 \\
0 & 0 &  \rr_3 & 0 & 0 &\rr_5  \\
\pos_1 & 0 &  0 & \rr_1 & 0 &0  \\
0 & \rr_4 & 0 & 0  &\rr_2 & 0 \\
0 & 0 & \rr_5 & 0 & 0 &\rr_3  
\end{pmatrix}
\end{eqnarray*}
for some $\alpha_1, \ldots, \alpha_5\in \setR$ and $\beta_1>0$.  
We may assume that $\alpha_4$ and $\alpha_5$ are non-zero since
otherwise  $\operatorname{span} \{dx^i, dx^2\}\subset
\cF\vert_p$ for some $i\in \{0,1\}$ as in Metaclass IV.
Then the Tamm-Rubilar tensor density satisfies
\begin{eqnarray*}
C^{-1} \, \cG^{ijkl} \xi_i \xi_j \xi_k \xi_l 
 &=& \xi_0^4+\xi_1^4-\xi_2^4-\xi_3^4 + D_0 \xi_0 \xi_1 \xi_2 \xi_3 \\
 & & \quad + D_{1} (\xi_2^2 \xi_3^2- \xi_0^2 \xi_1^2)
- D_{2} (\xi_1^2 \xi_3^2+ \xi_0^2 \xi_2^2)
- D_{3} (\xi_1^2 \xi_2^2+ \xi_0^2 \xi_3^2), 
\end{eqnarray*}
where $C =  {\beta_1 \alpha_4 \alpha_5}$
and  $D_0, D_1, D_2, D_3\in \setR$ are defined in terms of $\alpha_i$ and $\beta_1$. 
By decomposing the Fresnel tensor as in equation \eqref{eq:metaclassI-fDecomp} and eliminating variables 
using a Gr\"obner basis, it follows that there exists a $\sigma\in \{\pm 1\}$ such that
\begin{eqnarray*}
  D_0 = 0, \quad D_1 = \sigma 2, \quad D_2 = -\sigma D_3,
\end{eqnarray*}
and moreover, equation \eqref{eq:metaclassI-fDecomp} holds for $A=g_+^{-1}$, $B=g_{-}^{-1}$ 
and $C$ as above, where
\begin{eqnarray*}
 g_\pm &=& 
\operatorname{diag}\left(1,
-\sigma, 
\frac 1 2\left(\sigma D_3 \pm \sqrt{D_3^2+4}\right),
\frac 1 2\left(-D_3 \mp \sigma \sqrt{D_3^2+4}\right)
\right)^{-1}.
\end{eqnarray*}
Since $g_+$ does not have a Lorentz signature for any $\sigma\in \{\pm
1\}$ and $D_3 \in \setR$, Proposition \ref{thm:localDescGc} and unique factorisation imply that
$\kappa\vert_p$ can not be in Metaclass VI.

\textbf{Metaclass VII.}  If $\kappa\vert_p$ is in Metaclass VII, there
are coordinates $\{x^i\}_{i=0}^3$ around $p$ such that
\begin{eqnarray*}
(\kappa_I^J)_{IJ} &=& 
\begin{pmatrix}
\rr_1 & 0 & 0 & \rr_4 & 0 &0 \\
0 & \rr_2 & 0 & 0  &\rr_5 & 0 \\
0 & 0 &  \rr_3 & 0 & 0 &\rr_6  \\
\rr_4 & 0 &  0 & \rr_1 & 0 &0  \\
0 & \rr_5 & 0 & 0  &\rr_2 & 0 \\
0 & 0 & \rr_6 & 0 & 0 &\rr_3 
\end{pmatrix}
\end{eqnarray*}
for some $\alpha_1, \ldots, \alpha_6\in \setR$.  We may assume that
$\alpha_4, \alpha_5, \alpha_6\neq 0$ since otherwise
$\operatorname{span}\{dx^i\vert_p, dx^j\vert_p\}\subset \cF\vert_p$
for some $i,j\in \{0,1,2\}$ as in Metaclass IV.  Then the Tamm-Rubilar tensor density 
satisfies
\begin{eqnarray*}
C^{-1} \, \cG^{ijkl} \xi_i \xi_j \xi_k \xi_l 
 &=& \xi_0^4 + \xi_1^4 + \xi_2^4 + \xi_3^4 
+ D_0 \xi_0 \xi_1 \xi_2 \xi_3 - 
\sum_{i=1}^3 D_i (\xi_{i'}^2 \xi_{i''}^2 + \xi_0^2 \xi_i^2), 
\end{eqnarray*}
where $C = {\alpha_4 \alpha_5\alpha_6}$, constants $D_1, D_2, D_3\in \setR$ are given by
\begin{eqnarray}
\label{eq:VII-def-D}
 D_1 &=& \frac{ (\alpha_{2}-\alpha_{3})^2- \alpha^2_{5}-\alpha^2_{6}}{ \alpha_{5}\alpha_{6}},\\
 D_2 &=& \frac{ (\alpha_{1}-\alpha_{3})^2- \alpha^2_{4}-\alpha^2_{6}}{ \alpha_{4}\alpha_{6}},\\
 D_3 &=& \frac{ (\alpha_{1}-\alpha_{2})^2- \alpha^2_{4}-\alpha^2_{5}}{ \alpha_{4}\alpha_{5}},
\end{eqnarray}
and $D_0\in \setR$ is given explicitly in terms of $\alpha_1, \ldots, \beta_3$,
and implicitly $D_0$ satisfies 
\begin{eqnarray} 
\label{eq:MetaclassVII_D0}
 D_0^2 = 4\left( -4+ D_1 D_2 D_3 + D_1^2+ D_2^2+ D_3^2\right).
\end{eqnarray}
Decomposing the Tamm-Rubilar tensor density as in equation
\eqref{eq:metaclassI-fDecomp} and eliminating variables using a
Gr\"obner basis, gives polynomial equations for $D_0, D_1, D_2,
D_3$.  
Let us consider the cases $D_0=0$ and $D_0\neq 0$ separately.
If $D_0 = 0$, there exists an $i\in \{1,2,3\}$ and a $\sigma\in
\{\pm 1\}$ such that
\begin{eqnarray}
\label{eq:VII-caseD00}
   D_0 = 0, \quad
   D_i = -\sigma 2, \quad
   D_{i'} = \sigma D_{i''},
\end{eqnarray}
where the last condition is a consequence of equation
\eqref{eq:MetaclassVII_D0}.  Suppose $i=1$. 
Then Proposition
\ref{thm:localDescGc} implies that for some invertible symmetric
matrices $A,B\in \setR^{4\times 4}$ with Lorentz signatures we have
\begin{eqnarray}
\label{eq:VII-ssss}
\left(  \xi^t \cdot A\cdot \xi\right)\left(  \xi^t \cdot B\cdot \xi\right)&=&
\left(  \xi^t \cdot L_+\cdot \xi\right)\left(  \xi^t \cdot L_-\cdot \xi\right), \quad \xi\in \setC^4,
\end{eqnarray}
where matrices $L_\pm  \in \setC^{4\times 4}$ are defined as
\begin{eqnarray}
\label{eq:gpmVII}
 L_\pm &=& 
    \operatorname{diag}\left(
        1, 
        \sigma,
        \frac 1 2\left(-D_2 \pm \sqrt{D_2^2-4}\right),
        \frac \sigma 2\left(-D_2 \pm \sqrt{D_2^2-4}\right)
    \right).
\end{eqnarray}
Since $L_\pm$ are invertible, equation \eqref{eq:VII-ssss}, Proposition
\ref{lemma:irr} and unique factorisation imply that $L_\pm$ are real
and have Lorentz signatures.  Thus $\vert D_2\vert\ge 2$ and $\det L_\pm <0$, but this
contradicts equation \eqref{eq:gpmVII}, which implies that
\begin{eqnarray*} 
  \det L_\pm = \frac 1 4 \left( -D_2 \pm \sqrt{D_2^2-4}\right)^2 >0.
\end{eqnarray*}
A similar analysis for $i=2,3$ shows that the case $D_0=0$ is not possible.
If $D_0\neq 0$ it follows that there exists $\sigma_1, \sigma_2,
\sigma_3 \in \{\pm 1\}$ and distinct $i,j, k\in \{1,2,3\}$ such that
\begin{eqnarray}
\label{eq:metaVII-D0n0}
  D_0 \neq 0, \quad
  D_i = \sigma_1 2, \quad  
  D_j = \sigma_2 2, \quad  
  D_k = \frac 1 2 \left( -4 \sigma_1 \sigma_2 + \sigma_3 D_0 \right),
\end{eqnarray}
where the last equation follows from equation
\eqref{eq:MetaclassVII_D0}.  If $(i,j)=(1,2)$ then $k=3$ and
Proposition \ref{thm:localDescGc} implies that for some invertible
symmetric matrices $A,B\in \setR^{4\times 4}$ with Lorentz signatures,
equation \eqref{eq:VII-ssss} holds for matrices $L_\pm \in
\setC^{4\times 4}$ defined as
\begin{eqnarray}
\label{eq:VII-gppz}
 L_\pm &=& \begin{pmatrix}
    1 & 0 & 0 & \pm \frac{\sqrt{D_0}}{\sqrt{8} \sqrt{\sigma_3}} \\
   0  & -\sigma_1 & \mp \frac{\sqrt{D_0} \sqrt{\sigma_3} }{\sqrt{8}} & 0 \\
   0  &  \mp \frac{\sqrt{D_0}\sqrt{\sigma_3}}{\sqrt{8} } & -\sigma_2 & 0 \\
\pm \frac{\sqrt{D_0}}{\sqrt{8} \sqrt{\sigma_3}} &  0 & 0 & \sigma_1 \sigma_2
\end{pmatrix}.
\end{eqnarray} 
Since both sides in equation \eqref{eq:VII-ssss} should decompose into
the same number of irreducible factors, it follows that $\xi^t\cdot
L_\pm \cdot \xi$ are irreducible in $\setC[\xi_0, \ldots,
\xi_3]$. Thus equation \eqref{eq:VII-ssss} and unique factorisation
imply that $L_\pm$ are real and have Lorentz signatures, so $\det
L_\pm <0$. However, this contradicts equation \eqref{eq:VII-gppz}
which implies that
\begin{eqnarray*} 
  \det L_\pm = \left( \frac 1 8 D_0 - \sigma_1 \sigma_2 \sigma_3 \right)^2 \ge 0.
\end{eqnarray*}
The cases $(i,j) = (1,3)$, $(2,3)$ are excluded by the same argument by using metrics
\begin{eqnarray*}
 L_\pm &=& \begin{pmatrix}
    1 & 0 & \pm \frac{\sqrt{D_0}\sqrt{\sigma_3}}{\sqrt{8} } & 0 \\
   0  & -\sigma_1 & 0 & \mp \frac{\sqrt{D_0}  }{\sqrt{8}\sqrt{\sigma_3}} \\
    \pm \frac{\sqrt{D_0}\sqrt{\sigma_3}}{\sqrt{8} } & 0 & \sigma_1 \sigma_2 & 0 \\
0 & \mp \frac{\sqrt{D_0}}{\sqrt{8}\sqrt{\sigma_3}}  & 0 & -\sigma_2
\end{pmatrix}, \\
 L_\pm &=& \begin{pmatrix} 
    1 & \pm \frac{\sqrt{D_0}}{\sqrt{8} \sqrt{\sigma_3}}  & 0 & 0 \\
\pm \frac{\sqrt{D_0}}{\sqrt{8} \sqrt{\sigma_3}}  & \sigma_1\sigma_2 & 0 & 0 \\
  0& 0 & -\sigma_1 &   \mp \frac{\sqrt{D_0}\sqrt{\sigma_3}}{\sqrt{8} } \\
0 & 0 & \mp \frac{\sqrt{D_0}\sqrt{\sigma_3}}{\sqrt{8} }   & -\sigma_2
\end{pmatrix}, 
\end{eqnarray*}
respectively. Thus $\kappa\vert_p$ can not be in metaclasses VII.

\textbf{Metaclasses VIII---XXIII.} 
(Following \cite[Lemma 5.1]{Schuller:2010}.)
 Let $A=(\kappa_I^J)_{IJ}$ be the
$6\times 6$ matrix that represents $\kappa\vert_p$ in some coordinates
$\{x^i\}_{i=0}^3$ around $p$. Then the Jordan normal form of $A$ has a
block of dimension $d\in \{2,\ldots, 6\}$ that corresponds to a real
eigenvalue $\lambda\in \setR\slaz$.  By considering unit vectors in the normal basis,
we can find non-zero $e_1, e_2\in \Lambda^2(N)\vert_p$ so that $\kappa(e_1) = \lambda
e_1$ and $\kappa(e_2) = \lambda e_2 + e_1$. Writing out
$\kappa(e_1)\wedge e_2 = e_1\wedge \kappa(e_2)$ implies that
$e_1\wedge e_1=0$, so $e_1 = \eta_1\wedge \eta_2$ for some linearly
independent $\eta_1,\eta_2\in \Lambda^1(N)\vert_p$
\cite[p. 184]{Cohn:2005}.  Let $W= \operatorname{span} \{ \eta_1,
\eta_2\}$. For all $\xi\in W$ we then have
\begin{eqnarray*}
W &\subset& \{ \alpha\in \Lambda^1(N)\vert_p : \xi\wedge \kappa(\xi\wedge \alpha)=0\},
\end{eqnarray*}
whence Theorem 3.3 in \cite{Dahl:2011:Closure} implies that $ W
\subset \mathscr{F}\vert_p $ and Proposition
\ref{prop:noTwoDimSubspace} implies that $\kappa\vert_p$ can not be in metaclasses
VIII-XXIII.
\end{proof}

In the proof of Theorem \ref{thm:factorMedium} the assumption that
$\kappa$ is invertible is only used to show that $\alpha_3^2\neq
\alpha_4^2$ in Metaclass IV and to exclude Metaclasses
VIII---XXIII. It would be interesting to see if these last metaclasses
can be excluded also for non-invertible $\kappa$ by other arguments.
Regarding this question it should be emphasized that here $\kappa$ is
real. For complex coefficients $\kappa$, the setting becomes more
involved.  For example, in \cite[Example 5.3]{Dahl:2011:Closure} it is
shows that for complex $\kappa$ the Fresnel surface can be a single
light cone even if $\kappa$ is not invertible.  Also, \emph{chiral
  medium} would be an example of a medium with complex coefficients in
$\kappa$ and with a double light cone.  In chiral medium right and
left hand circularly polarised plane waves propagate with different
wavespeeds.  Let us note that if we set $\xi_0=0$ in 
equation \eqref{eq:fresnel:metaI} we obtain the ternary quartic
studied in \cite{Thomsen:1916} and for this polynomial, $D_0^2$ in
equation \eqref{eq:MetaclassI_D0} is one of the factors in the
discriminant.

Let us make some comments regarding the three mediums derived
in Theorem \ref{thm:factorMedium}.  A first observation is that for
each Metaclass in Theorem \ref{thm:factorMedium}, the light cones are
parameterised by only one parameter: $D_{i'}=D_{i''}\ge 2$ in Metaclass I,
$\beta_1>0$ in Metaclass II, and $D_1\in \setR$ in Metaclass IV. Let
consider each metaclass under the assumption that Theorem
\ref{thm:factorMedium} holds.

\textbf{Metaclass I.}
In Metaclass I we assume that conditions \eqref{eq:mc-I-Dsol} holds for only one $i\in \{1,2,3\}$.
In terms of $\alpha_1, \ldots, \beta_3$ conditions \eqref{eq:mc-I-Dsol} 
are equivalent to
\begin{eqnarray*}
\label{eq:MetaclassI-cond}
  \alpha_{i'} = \alpha_{i''}, \quad \beta_{i'} = \beta_{i''}.
\end{eqnarray*}
When $\alpha_1 = \alpha_2 = \alpha_3 = 0$ this medium reduces to a uniaxial medium, where wave propagation
is well understood.

Let us also note that if $D_i=D_j=2$ for two distinct $i,j\in
\{1,2,3\}$, then $\alpha_1 = \alpha_2= \alpha_3$ and $\beta_1 =
\beta_2=\beta_3$. Then $\kappa\vert_p = -\beta_1 \ast_g + \alpha_1
\operatorname{Id}$ for the locally defined Lorentz metric $g =
\operatorname{diag}(-1,1,1,1)$ and the Fresnel surface is the single
light cone
\begin{eqnarray*}
\label{eq:FreEqN}
  \cF\vert_p &=& N_p(g).
\end{eqnarray*}

\textbf{Metaclass II.} 
For Metaclass II, the Fresnel polynomial
$\cG^{ijkl}\xi_i \xi_j \xi_k \xi_l$ is a function of $\xi_0, \xi_1^2+\xi_2^2, \xi_3$.
It is therefore motivated to project $\cF\vert_p$ onto $\xi_1=0$, and  we can plot $\cF\vert_p$
as a surface in $\setR^3$. Figure \ref{fig:ConesII} show this projection for 
three different values of $\beta_1$. From the figures (or from metrics $g_\pm$) we see that
the light cones $N(g_\pm)$ coincide in the limit $\beta_1\to \infty$.

\begin{center}
\begin{figure}[!ht]
\includegraphics[width=0.99\textwidth]{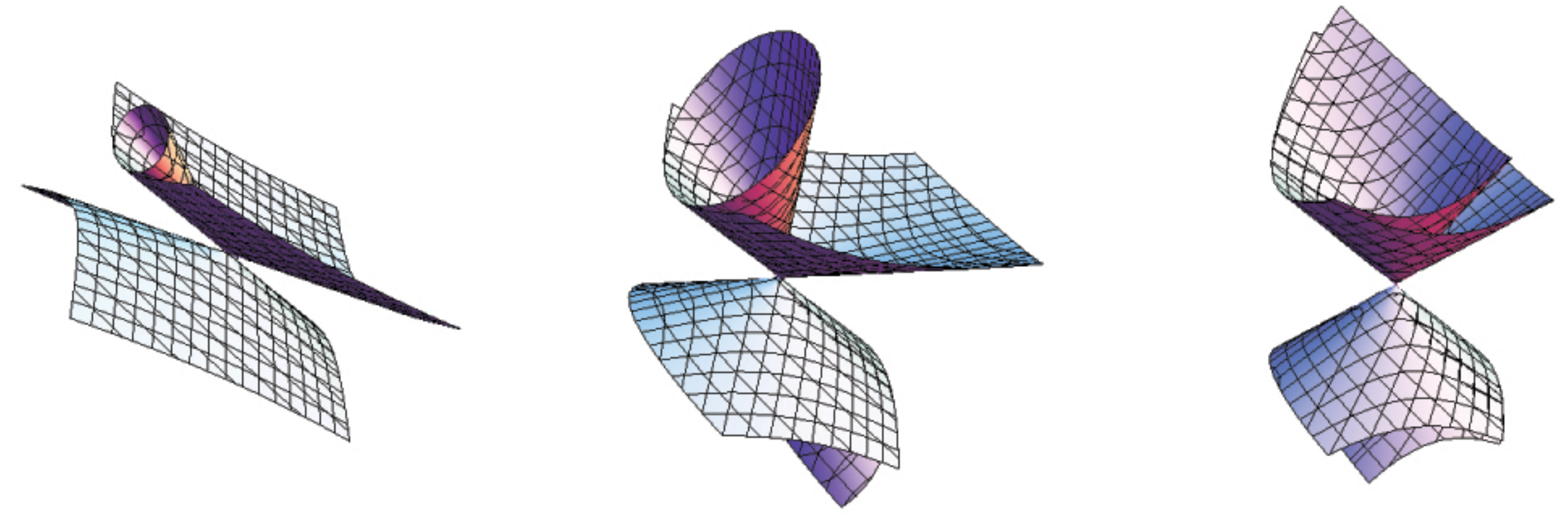}
\caption{
Projection into $\setR^3$ of Fresnel surfaces in Metaclass II
for $\beta_1 = 0.2$ (left), $\beta_1 = 0.8$ and $\beta_1 = 5$ (right).}
\label{fig:ConesII}
\end{figure}
\end{center}

\newcommand{\qw}[0]{w}

By a coordinate transformation we can put the local representation of $\kappa\vert_p$ into a 
more symmetric form.
Let  $\{\widetilde x^i\}_{i=0}^3$ be coordinates defined as
$\widetilde x^i = \sum_{j=0}^3 L_{ij} x^j$ where $L=(L_{ij})$ is the Jacobian matrix
\begin{eqnarray*}
L &=& \begin{pmatrix}
  0 & 0 & \frac{1}{2\beta_1}(1-\qw) & \frac{1}{2\beta_1}(1+\qw) \\
0 & 1 & 0 & 0\\
1 & 0 & 0 & 0\\
0& 0& 1 & 1
\end{pmatrix}^{-1},
\end{eqnarray*}
where $\qw = \sqrt{1+4\beta_1^2}$. The motivation for these coordinates is
that they diagonalize $g_+$. 
Then the $6\times 6$ matrix $(\widetilde \kappa_I^J)_{IJ}$ that
represents  $\kappa\vert_p$ in $\{\widetilde x^i\}_{i=0}^3$ coordinates is 
\begin{eqnarray*}
\alpha_1 \operatorname{Id} + \frac 1 {\qw} \begin{pmatrix}
0 & 0 & 0 & \beta_1^2  & 0 &0 \\
0 & \beta_1 & -\beta_1 & 0  &\beta_1(-1+\qw) & -\beta_1 \\
0 & \beta_1 &  -\beta_1 & 0 & -\beta_1 &-\beta_1(1+\qw)  \\
-\qw^2 & 0 &  0 & 0 & 0 &0  \\
0 & -\beta_1(1+\qw) & \beta_1  & 0  &\beta_1 & \beta_1 \\
0 & \beta_1 & \beta_1(-1+\qw) & 0 & -\beta_1 &-\beta_1
\end{pmatrix}.
\end{eqnarray*}

\subsection{Metaclass IV}
For Metaclass IV, the Fresnel polynomial is also a function of 
 $\xi_0, \xi_1^2+\xi_2^2, \xi_3$. We may therefore visualize the 
Fresnel surface in the same way as in Metaclass II. See Figure \ref{fig:ConesIV}.

\begin{center}
\begin{figure}[!ht]
\includegraphics[width=0.99\textwidth]{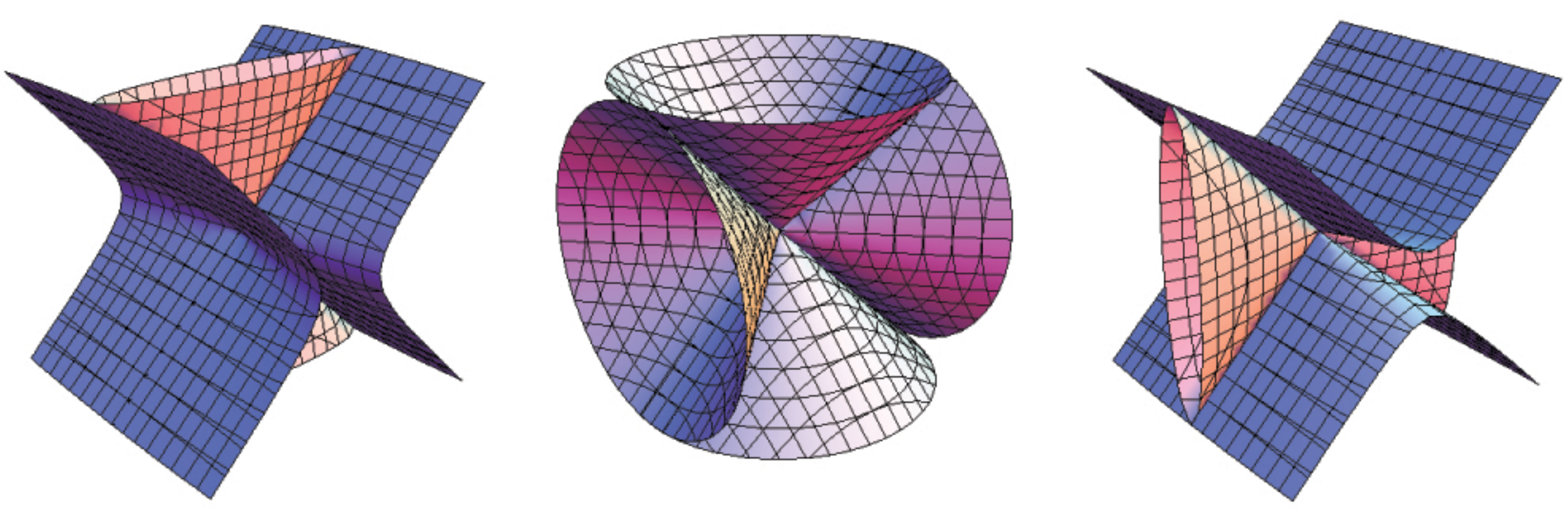}
\caption{
Projection into $\setR^3$ of Fresnel surfaces in Metaclass IV
 for $D_1 = -25$ (left), $D_1 = 0$ and $D_1 = 25$ (right).}
\label{fig:ConesIV}
\end{figure}
\end{center}

Suppose $\alpha_1=\alpha_2=\alpha_3=0$. If we treat $x^0$ as time and
$\{x^i\}_{i=1}^3$ as space coordinates, and write Maxwell's equations using vector fields 
$\mathbf{E}, \mathbf{D}, \mathbf{B}, \mathbf{H}$, then $\kappa\vert_p$ represents medium 
\begin{eqnarray*}
  \mathbf{D} &=& -\operatorname{diag}(\beta_1, \beta_1, \alpha_4) \cdot \mathbf{E}, \\
  \mathbf{B} &=& -\operatorname{diag}(\beta_1, \beta_1, -\alpha_4)^{-1} \cdot \mathbf{H}.
\end{eqnarray*}

For an electromagnetic medium to be physically relevant, the Fresnel
polynomial $\cG^{ijkl}\xi_i \xi_j \xi_k \xi_l$ should be
\emph{hyperbolic polynomial} \cite{Schuller:2010}.  This is a necessary
condition for Maxwell's equations to form a predictive theory, that
is, a necessary condition for Maxwell's equations to be solvable
forward in time.  The above 3D projections and the argument in
\cite{Schuller:2010} suggests that Metaclass II is hyperbolic for all
$\beta_1>0$ while Metaclass IV is never hyperbolic for any $D_1\in
\setR$. This is also supported by some  numerical tests.

\textbf{Acknowledgements.} 
The author gratefully appreciates financial
support by the Academy of Finland (project 13132527 and Centre of
Excellence in Inverse Problems Research), and by the Institute of
Mathematics at Aalto University.

\providecommand{\bysame}{\leavevmode\hbox to3em{\hrulefill}\thinspace}
\providecommand{\MR}{\relax\ifhmode\unskip\space\fi MR }
\providecommand{\MRhref}[2]{%
  \href{http://www.ams.org/mathscinet-getitem?mr=#1}{#2}
}
\providecommand{\href}[2]{#2}

\end{document}